\documentclass[12pt]{article}

\usepackage{amsmath}
\usepackage{amssymb}
\usepackage{amsthm}
\usepackage{mathrsfs}
\usepackage{tikz}

\usepackage{slashed}
\usepackage{color}
\usepackage{esint}

\usepackage{graphicx}
\usepackage[all,cmtip]{xy}
\usepackage{hyperref}

\newtheorem{thm}{Theorem}
\newtheorem{lem}{Lemma}
\newtheorem{prop}{Proposition}
\newtheorem{cor}{Corollary}

\newtheorem{defn}{Definition}
\newtheorem{examp}{Example}

\newtheorem{rmk}{Remark}

\hypersetup{
    pdfborder = {0 0 0},%
    colorlinks,%
    citecolor=blue,%
    filecolor=black,%
    linkcolor=red,
    urlcolor=green
}

\begin{document}

\title{ Bott--Kitaev periodic table and index theory }
\author{ Dan Li  \thanks{ email: li1863@math.purdue.edu }  
	    }

\date{}
\maketitle

\begin{abstract}
 We consider topological insulators and superconductors with discrete symmetries 
 and  clarify the relevant index theory behind the periodic table proposed by Kitaev. 
 An effective Hamiltonian determines the analytical index, which can be computed by a topological index.
 We focus on the spatial dimensions one, two and three, and only consider the bulk theory. 
 In two dimensions, the $\mathbb{Z}$-valued invariants are given by 
 the first Chern number. Meanwhile, $\mathbb{Z}_2$-valued invariants can be computed by the odd topological index and its variations. 
 The Bott-Kitaev periodic table is well-known in the physics literature, we organize the topological invariants in the framework of KR-theory.

\end{abstract}

\section{Introduction}\label{Intro}

The concept of symmetry is the foundation for many fundamental theories of modern physics. For example,
 Wigner's theorem  places severe restrictions on  possible operators representing symmetry transformations in quantum mechanics.
More precisely, if $\mathcal{H}$ is a complex Hilbert space of physical states and $S:\mathcal{H} \rightarrow \mathcal{H} $ is a norm-preserving surjective symmetry transformation, then $S$ is a unitary or anti-unitary  operator.

In the study of topological phases of matter, one considers various discrete symmetries such as the time reversal symmetry (TRS), particle hole symmetry (PHS) and chiral symmetry (CS) 
(sometimes also called sub-lattice  symmetry (SLS)). Such $\mathbb{Z}_2 $ symmetries have significant effects on the band structure of   topological insulators and superconductors 
(sometimes also called  generalized topological insulators). 
More precisely, one considers the topological band theory of free fermionic ground states of gapped systems with discrete $\mathbb{Z}_2 $  symmetries, and models their band structures by topological vector bundles. As a classification
scheme, topological K-theory can be used to classify topological band structures.  Those $\mathbb{Z}_2$ symmetries introduce involutions (or anti-involutions) on vector bundles and the base manifold. So
the Real K-theory  (i.e., KR-theory \cite{A66})  provides a perfect model to study the topological band theory for topological insulators and superconductors.

Dyson's three-fold way \cite{D62} shows that the orthogonal, unitary and symplectic groups are respectively fundamental symmetries in real, complex and quaternionic quantum mechanics. The Wigner--Dyson symmetry classes
are later generalized to Altland--Zirnbauer symmetry classes \cite{AZ97, HHZ05, Z96}, which is  referred to as the ten-fold way. 
That is, the combinations of three $\mathbb{Z}_2$ symmetries  TRS, PHS and CS give ten symmetry classes in total.  
Since its intimate relation to Cartan's classification of symmetric spaces \cite{SCR11}, the tenfold way is sometimes also called the Cartan--Altland--Zirnbauer classification scheme.

Kitaev  observed that the table of possible phases of topological insulators and superconductors  resembles the Bott periodicity in K-theory \cite{K09}, 
and then  proposed the K-theoretic classification of  generalized topological insulators with $\mathbb{Z}_2$ symmetries. 
Indeed, the Altland--Zirnbauer symmetry classes
induce involutions or anti-involutions in the topological band theory, and such involutive vector bundles fall into the category of KR-theory. 
The relevant topological phase (or invariant) is expected to be an index number with K-theory as the receptacle. 
Given the effective Hamiltonian of a generalized topological insulator,  one first considers its analytical index, and then find a  topological index formula to compute the topological invariant. 
In other words, the topological invariant is understood as an index theorem in KR-theory.
For instance, the topological $\mathbb{Z}_2$ invariant of time reversal invariant topological insulator is a mod 2 index theorem in KR-theory \cite{KLW15},
which is the prototypical example of a $\mathbb{Z}_2$ invariant counting the parity of Majonara zero modes.

It is quasi-particles (rather than real particles) that are the fundamental objects manifesting rich topological phases in condensed matter physics.
For example, Majorana zero modes in superconductors consist of an electron and a hole intersecting at the zero energy,
and Majorana bound states of time reversal invariant topological insulators give rise to Dirac cones.  
It is believed that Majorana bound states or Majorana zero modes (both will be called Majonara zero modes due to similar properties)  will have promising applications in designing quantum computers \cite{DFN15}.
We point out that all $\mathbb{Z}_2$ invariants are index numbers in $KO^{-2}(pt) = \mathbb{Z}_2$ as a feature of quasi-particles.

We study generalized topological insulators and relevant topological invariants in the framework of index theory and K-theory.
With involutions or anti-involutions induced by $\mathbb{Z}_2$ symmetries, the topological band theory falls into the category of Real K-theory, and we 
emphasize that the bulk topological invariants are index numbers in KR-theory. The diagonal map in the periodic table is a consequence of the $(1,1)$-periodicity of KR-theory. 
So a topological invariant can be computed as a topological index in the bulk, which will be collected
in a table as the main result.
The boundary theory will not be discussed in this paper, nor the bulk-boundary correspondence.

This paper collects examples of index theorems in the literature, and organizes them in the unified framework of KR-theory.   
For example, the mod 2 index theorem of time reversal invariant topological insulators concerns the parity
of Majorana zero states  and computes the conductance as a topological index \cite{KLW15}.
Index theorems for Majorana zero modes in topological superconductors have been established, for example, in \cite{FF10, R10, TDL07}. 
In a modern language, an index number can be obtained as the index pairing between K-homology and K-theory, this approach has been discussed in \cite{GS15}.
For more details about the periodic table in physics, the reader can consult the review papers \cite{HK10, QZ11, SA17, SRFL09}.

This paper is organized as follows. Sect. \ref{BG} gives some background about the Bott--Kitaev periodic table, KR-theory and index theory. 
Sect. \ref{sec:TRS} mainly reviews the mod 2 index theorem in KR-theory for time reversal invariant topological insulators, i.e., type AII.
Sect. \ref{sec:PHS} gives plenty of examples of topological superconductors and the relevant topological invariants. 
Sect. \ref{sec:CS} gives examples in the chiral orthogonal and symplectic classes. 
Sect. \ref{main} contains the main observations from the above examples and provides a table of topological indices for practical computations.

\section{Background}\label{BG}

In this section, we  briefly review some basic facts about  $\mathbb{Z}_2$ symmetries in generalized topological insulators and the Bott--Kitaev periodic table, which is the main object we want to clarify using index theory \cite{AS69}.
The relevant topological band theory can be studied by the Real K-theory, i.e., KR-theory \cite{A66}, and the topological invariants will be understood as index numbers.

\subsection{Periodic table}

The $\mathbb{Z}$- and $\mathbb{Z}_2$-valued topological  invariants of topological insulators and Bogoliubov-de Gennes (BdG) superconductors
fit into a periodic table resembling the Bott periodicity as in topological K-theory. A Clifford algebra classification was
established in  \cite{SRFL09} and the K-theory classification followed as the Clifford extension problem \cite{K09}.
A finer classification of topological phases based on twisted equivariant K-theory was discussed in \cite{FM13}. 

Given an effective Hamiltonian $H$ parameterized by the momentum space $X$ (or $\mathbf{k}$-space), one considers three $\mathbb{Z}_2$ symmetries: the time reversal symmetry $\mathcal{T}$, particle hole symmetry $\mathcal{C}$ and chiral symmetry $\mathcal{S}$.
If $H$ is time reversal invariant, which means that $\mathcal{T} H \mathcal{T}^{-1} = H$,  and additionally $\mathcal{T}^2 =\pm 1$ depending on the spin being integer or half-integer,
 the time reversal symmetry (TRS) introduces three classes
\begin{equation*}
    TRS = \left\{
  \begin{array}{l l}
    +1 & \quad \text{if} \quad \mathcal{T} H(\mathbf{k}) \mathcal{T}^* = H(- \mathbf{k}), \,\, \mathcal{T}^2=+1\\
    -1 & \quad \text{if} \quad \mathcal{T} H (\mathbf{k})\mathcal{T}^* = H(-\mathbf{k}), \,\, \mathcal{T}^2 =-1 \\
    0  & \quad \text{if} \quad \mathcal{T} H(\mathbf{k}) \mathcal{T}^* \neq H(-\mathbf{k})
  \end{array}  \right.
\end{equation*}
Similarly, the particle hole symmetry (PHS) also gives three classes
\begin{equation*}
    PHS = \left\{
  \begin{array}{l l}
    +1 & \quad \text{if} \quad \mathcal{C} H(\mathbf{k}) \mathcal{C}^* =- H(-\mathbf{k}), \,\, \mathcal{C}^2=+1\\
    -1 & \quad \text{if} \quad \mathcal{C} H(\mathbf{k}) \mathcal{C}^* = -H(-\mathbf{k}), \,\, \mathcal{C}^2 =-1 \\
    0  & \quad \text{if} \quad \mathcal{C} H(\mathbf{k}) \mathcal{C}^* \neq -H(-\mathbf{k})
  \end{array}  \right.
\end{equation*}
The chiral symmetry (CS) is defined as the product $\mathcal{S} = \mathcal{T} \cdot \mathcal{C}$, sometimes also referred to as the sublattice symmetry (SLS). 
If both $\mathcal{T}$ and $\mathcal{C}$ are nonzero, then the chiral symmetry is present, i.e., $\mathcal{S} = 1$.

The unitary class or type A is the symmetry class that has no TRS, PHS or CS/SLS.
The chiral unitary class or type AIII is the symmetry class that has the chiral symmetry but no TRS or PHS.
These two classes can be classified by complex K-theory, i.e., $K^{i}(X)$ ($i = 0, 1$), which will not be investigated in this paper. 
In fact, we only focus on those eight classes related to KR-theory. 
The following is the periodic table of topological insulators and superconductors in spatial dimensions one, two and three, 
we will look into the topological invariants and interpret them as index numbers.  

\begin{center}
 \begin{tabular}{||c| c| c| c | c | c | c | c |c |c ||}
 \hline
    Type & TRS & PHS & CS/SLS &  Class & d=1 & d=2 & d=3  \\ [0.5ex]
 \hline\hline
  AI & +1 & 0 & 0  & orthogonal & 0 &  0  & 0 \\
  \hline
 BDI & +1 & +1 & 1   & chiral orthogonal & $\mathbb{Z}$ & 0  & 0 \\
 \hline
 D & 0 & +1 & 0  &  BdG & $\mathbb{Z}_2$ &  $\mathbb{Z}$  & 0 \\
  \hline
 DIII & $-1$ & +1 & 1  & BdG  & $\mathbb{Z}_2$ & $\mathbb{Z}_2$  & $  \mathbb{Z}$  \\  
 \hline
  AII & $-1$ & 0 & 0  & symplectic & 0 & $\mathbb{Z}_2$ &  $\mathbb{Z}_2$   \\
  \hline
CII & $-1$ & $-1$ & 1   & chiral symplectic & $\mathbb{Z}$ & 0 & $\mathbb{Z}_2$  \\
 \hline
 C & 0 & $-1$ & 0  & BdG & 0 & $\mathbb{Z}$ &  0 \\
  \hline
 CI & +1 & $-1$ & 1   & BdG & 0  & 0  & $\mathbb{Z} $ \\ [1ex]
 \hline
\end{tabular}
\end{center}

\subsection{KR-theory}

Let $X$ be a compact space,  that is used to model the momentum space of the lattice model of  a topological insulator.
\begin{examp}
A lattice in $\mathbb{R}^d$ is a free abelian group isomorphic to $\mathbb{Z}^d$ and its Pontryagin dual is the torus $\mathbb{T}^d$. 
When one considers the translational symmetry on a lattice,
 one has the momentum space   $X = \mathbb{T}^d$.

The limit of a lattice model is the continuous model defined on $\mathbb{R}^d$ and its Pontryagin dual is itself,
in this case the momentum space is the one point compactification of $\mathbb{R}^d$, i.e.,  $X = (\mathbb{R}^d)^+ = \mathbb{S}^d$.
\end{examp}

\begin{defn}
An involutive space $(X, \tau)$ is a compact space $X$ equipped with an involution, i.e., a homeomorphism $\tau: X \rightarrow X$
such that $\tau^2 = id_X$.
\end{defn}
The pair $(X, \tau)$ is also called a Real space with the real structure $\tau$. The prototypical example of an involution
$\tau$ is given by the complex conjugation.

\begin{examp}
 Let $\mathbb{R}^{p, q}$ denote $\mathbb{R}^{p} \oplus i \mathbb{R}^{q} $,
the canonical involution $\tau$ on  $\mathbb{R}^{p, q}$ is defined
by the complex conjugation so that $\tau|_{\mathbb{R}^{p}} = +1$ and $\tau|_{i \mathbb{R}^{q}} = -1$.
By definition, $\mathbb{S}^{ p,q}$ denotes the unit sphere in $\mathbb{R}^{p,q}$, and in this notation the torus $\mathbb{T}^d = (\mathbb{S}^{ 1,1})^d$.
\end{examp}

The fixed points of an involution $\tau$ is the set,
\begin{equation*}
   X^\tau := \{ {x} \in X ~|~ \tau( {x} ) =  {x} \}
\end{equation*}
In other words, $X^\tau$ is the set of real points with respect to  $\tau$. We always assume that
$X^\tau$ is a finite set. If $X$ is a CW-complex, then the $\mathbb{Z}_2$ action induced by $\tau$ makes $X$ into a $\mathbb{Z}_2$-equivariant CW-complex.

The band structure of a topological insulator defines a topological vector bundle $\pi: E \rightarrow X$ 
over the base space $X$, i.e.,  the topological band theory. 
All possible topological vector bundles can be classified by K-theory, 
so topological K-theory can be used to classify topological band structures.

 A Hilbert bundle  is a vector bundle equipped with a Hermitian metric. 
 It is natural for a complex vector bundle to equip a Hermitian metric $g$, and up to isomorphism the choice of such a Hermitian metric is unique. 
 A Hilbert bundle (with a flat connection) is commonly  used to model continuous fields of Hilbert spaces in geometric quantization.
 Hence the space of sections of the Hilbert bundle $(E, g)$, denoted by $\Gamma(X, E)$,  models the Hilbert space of physical states of a topological insulator.

  A discrete symmetry defines a symmetry transformation on the momentum space, which can be lifted to the  vector bundle.
\begin{defn}
 A Real vector bundle $\pi: (E, \iota) \rightarrow (X, \tau)$ is a complex vector bundle $E$  equipped  with an
 involutive anti-linear bundle isomorphism  $\iota: E \rightarrow E$
 such that $\iota^2 = id_E$, and these two involutions are compatible, i.e., $\pi \circ \iota = \tau \circ \pi$.
\end{defn}

In the above definition, if the involution $\iota$ (s.t. $\iota^2 = 1$) is replaced by an anti-involution $\chi$
(s.t. $\chi^2 = -1$), then the pair $(E, \chi)$ defines a Quaternionic vector bundle over $(X, \tau)$.

\begin{defn} For a compact Real space $(X, \tau)$, the Real K-group $KR(X, \tau)$
is defined to be the Grothendieck group of finite rank Real vector bundles $\pi: (E, \chi) \rightarrow (X, \tau)$.
\end{defn}

If the involution $\tau$ is understood, then it is always omitted in the notation, i.e.,  $KR(X) = KR(X, \tau)$.
Higher KR-groups are defined by
$$
KR^{-p, -q}(X) := KR(X \times \mathbb{R}^{p, q})
$$
There exists an isomorphism, i.e., $(1,1)$-periodicity of KR-theory,
$$
KR^{p+1, q+1}(X) \cong KR^{p,q }(X)
$$
so by convention a KR-group is denoted by $KR^{p-q}(X) =KR^{p,q }(X) $. The Bott periodicity of KR-theory is 8,
$$
KR^{n+ 8 }(X) \cong KR^{n}(X)
$$
As usual, $\widetilde{KR}^n(X)$ denotes the reduced KR-group.

The  Quaternionic K-group $KQ(X, \tau)$ is similarly defined as the Grothendieck group of Quaternionic vector bundles $(E, \chi)$ over $(X, \tau)$.
There exists a canonical isomorphism between KR-theory and KQ-theory,
$$
KQ^n(X) \cong KR^{n-4}(X)
$$

When the involution $\tau$ in KR-theory is trivial, it becomes the real K-theory over $\mathbb{R}$, i.e., KO-theory, 
$$
KR^n(X, \tau = id) = KO^n(X)
$$
which also has periodicity 8: $KO^{n}(X) \cong KO^{n+8}(X)$.
The KO-theory of a point, for simplicity denoted as $KO^{n} = KO^n(pt)$, is given by the table

\begin{center}
 \begin{tabular}{||c| c| c| c | c | c | c | c |c |c |c ||}
 \hline
   i & 0 & 1 & 2 &  3 & 4 & 5 & 6 & 7  \\ [0.5ex]
 \hline\hline
 $KO^{-i}$ & $\mathbb{Z}$ & $\mathbb{Z}_2$ & $\mathbb{Z}_2$   & 0 & $\mathbb{Z}$  & 0  & 0 & 0 \\ [1ex]
 \hline
\end{tabular}
\end{center}

\subsection{Index theory} 

When a topological invariant belongs to a K-group, it is expected to be an index number.  
Furthermore, when the topological invariant can be computed as a topological index, an index theorem is expected. In the physics literature,
there are many examples of topological insulators and superconductors, and the relevant invariants can be computed as Chern numbers or winding numbers. 
After the work of Kitaev, the mathematical physics community believes index theory and K-theory should play a significantly important role in topological phase of matters.

As a milestone in  mathematics and physics, the Atiyah--Singer index theorem  
states that for an elliptic differential operator on a compact manifold the analytical index   is equal to the topological index.
In its modern form, the index theorem is a statement in K-theory. 
Given an effective Hamiltonian, one is interested in its  band structure.  
On the one hand, one considers the zero modes of the given Hamiltonian, whose dimension basically gives the analytical index. 
On the other hand, one models the band structure by a topological vector bundle, i.e., topological band theory, so one obtains an element in K-theory, 
which is the natural receptacle of the topological index. Through the spectral flow, one makes a connection between the analytical index and topological index. 
To prove the Atiyah--Singer index theorem using K-theory, one needs to reduce an elliptic operator to a Dirac operator, and make use of the Bott-periodicity and Thom isomorphism. 
Based on the Atiyah--Singer index theorem, the analytical index can be computed by the topological index. 

Given an effective Hamiltonian $H$ (equipped with symmetries), which is always assumed to be Fredholm,  one can define the Fredholm index as the analytical index, e.g., 
$$
ind_a(H) = \dim ker H - \dim coker H 
$$
Depending on the given symmetries, $H$ falls into different subclasses of Fredholm operators. The classifying spaces of KR-theory 
was discussed in Atiyah and Singer's original paper \cite{AS69}, later mentioned by Lott in \cite{L88},
and recently reformulated in the context of topological insulators in \cite{GS15}. 

The topological index maps from K-theory of the cotangent bundle $T^*X$ to that of an abstract point, in KR-theory  it is written as
$$
ind_t: KR^n(T^*X) \rightarrow KO^{n}(pt)
$$
For a $d$-dimensional involutive space $(X, \tau)$, the Thom isomorphism in $KR$-theory is  
\begin{equation*}
   KR^{-i}(X) \cong KR^{d-i}(T^*X)
\end{equation*}
Combining the above two maps, one obtains a map from KR-groups of $X$ to that of a point, still called  the topological index map,
$$
ind_t:  KR^{-i-d}(X) \rightarrow KO^{-i}(pt)
$$

For a local topological index formula, the even and odd cases should be treated differently. For example, in two dimensions the first Chern number (as a topological index) is the integral of the first Chern character
$$
c_1 = ind_t(p) = \frac{1}{2\pi}\int_X tr(pdpdp) = \frac{1}{2\pi} \int_X ch_1(p)
$$
where $p$ is a projection representing the relevant vector bundle.
In three dimensions,  the winding number (as an odd topological index) is the integral of the odd Chern character  
$$
ind_t(g) = \frac{1}{4\pi^2} \int_X tr(g^{-1}dg)^3 = \frac{1}{4\pi^2} \int_X ch_3(g)
$$
where $g: X \rightarrow U(m)$ is a gauge whose K-theoretic class $[g]$ is in an odd K-group. 
The odd index theorem is intimately related to spectral flow. 
The Chern character is a map from K-theory to de-Rham cohomology, in a general case, one has to include the defective class such as the Todd class or A-hat genus in the integral. 

\section{Time reversal symmetry}\label{sec:TRS}

In this section, we will discuss about the  time reversal  symmetry  $\mathcal{T}$ and topological phases of time reversal invariant  topological insulators: type AII and AI.
Type AII gives an important class of $\mathbb{Z}_2$ topological insulators, whose index theorem has been extensively studied in \cite{KLW15}. Type AI is not so interesting
since there is no nontrivial invariants in physics.

By definition,  time reversal symmetry (TRS) is the  $\mathbb{Z}_2$ symmetry  that reverses the direction of time.
 $\mathcal{T}$ changes the sign of the imaginary unit $\mathcal{T}:   i \mapsto -i$, so the time reversal operator  $\Theta$ (representing $\mathcal{T}$)   must be anti-unitary.
 The time reversal transformation defined
 on momentum spaces basically changes the sign of local coordinates.
\begin{defn}
 Time reversal symmetry defines the time reversal transformation on the momentum space $X$,
 \begin{equation*}
   \tau: X \rightarrow X; \quad   {x} \mapsto -{x}
 \end{equation*}
so that $(X, \tau)$ is an involutive space.
\end{defn}

Given a single-particle Hamiltonian $H(x)$ parameterized by the momentum space $X$, $H$ is time reversal invariant if
 \begin{equation}\label{TRIH}
   \Theta H(x) \Theta^* = H(\tau(x)), \quad \forall \,\, x \in X
 \end{equation}
Let us consider the  Hilbert bundle $\pi: \mathcal{H} \rightarrow X$, which models continuous fields of Hilbert spaces of physical states.
Time reversal symmetry is represented by the time reversal operator
 $\Theta$ acting on the Hilbert bundle $\mathcal{H}$. In general,
  $\Theta$ can be expressed as
 a product  $\Theta = UK$, where $U$ is a unitary operator and $K$ is the complex conjugation.
$\Theta$ is  an anti-linear anti-unitary operator, that is, for $\psi, \phi \in \Gamma(X, \mathcal{H})$,
\begin{equation*}
  \langle \Theta \psi, \Theta \phi \rangle = \overline{\langle  \psi,  \phi  \rangle} = \langle  \phi, \psi \rangle, \quad \quad \Theta (a \psi + b \phi) = \bar{a} \Theta\psi + \bar{b} \Theta \phi, \,\,\,\, a, b\in \mathbb{C}
\end{equation*}

 \begin{examp}
    In a spinless two-band model, the time reversal operator $\Theta = UK$ with $U := i\sigma_y$, where $\sigma_y$ is the second Pauli matrix. 
    In a spinful two-band model, $U$ is defined by $U :=-i\sigma_y  \otimes \tau_x$, where $\tau_x$ is the first Pauli matrix representing spin.
 \end{examp}

Taking  time reversal symmetry into account, the Hilbert bundle  $\pi: \mathcal{H} \rightarrow X$ becomes either
a Quaternionic Hilbert bundle $\pi: (\mathcal{H}, \Theta) \rightarrow (X, \tau)$ if $\Theta^2 = - id_\mathcal{H}$ or  a Real Hilbert bundle if $\Theta^2 = id_\mathcal{H}$.

\subsection{Type AII}

Consider the odd time reversal symmetry such that $\mathcal{T}^2 = -1$, denoted by TRS$-1$,
which can be found for example in the quantum spin Hall effect \cite{KM05}. 
In the Cartan--Altland--Zirnbauer classification, time reversal invariant  topological insulators with odd TRS fall into type AII, also called the symplectic class.  
This subsection reviews the index theorem of type AII topological insulators  following \cite{KLW15}.

If a Hamiltonian $H$ is time reversal invariant and the time reversal operator $\Theta$ is odd, 
$$
\Theta H(x) \Theta^* = H(\tau(x)), \quad \Theta^2 = -1
$$
then every energy level of $H$ is doubly degenerate. Namely, if $\phi$ is an eigenstate of $H$, then $\Theta \phi$ is an
orthogonal eigenstate of $H$ with the same energy.
This is  the Kramers degeneracy  for a   system in Type AII.

For simplicity, we assume the Hilbert bundle  $\pi: (\mathcal{H}, \Theta) \rightarrow (X, \tau)$ has rank 2. 
In this case, the Hilbert bundle is a Quaternionic vector bundle $\pi: (\mathcal{H}, \Theta) \rightarrow (X, \tau)$ s.t. $\Theta^2 = -1$. 
The transition function of $\mathcal{H}$, denoted by $w : X \rightarrow U(2)$, can be used to represent the time reversal symmetry. 

\begin{lem} \label{LemTRS}
 There exists an involution $\rho$  on the structure group $U(2)$ induced by the time reversal symmetry.
 
\end{lem}
\begin{proof}
   Choose an open subset $O \subset X$,
then the local isomorphism $\Theta : \mathcal{H}|_{O } \rightarrow \mathcal{H}|_{\tau(O) } $ is represented by
\begin{equation*}
   \Theta:  O \times \mathbb{C}^2 \rightarrow  \tau(O)  \times \mathbb{C}^2, \quad (x, v) \mapsto (\tau({x}), w(x) \, \bar{v})
\end{equation*}
Apply $\Theta$ twice to get back to $O$,
\begin{equation*}
   (x, v) \mapsto (\tau({x}), w(x) \, \bar{v}) \mapsto (x, w(\tau({x}))\bar{w}(x) \, {v})
\end{equation*}
because of $\Theta^2 = -1$, we have
\begin{equation}
   w(\tau({x}))\bar{w}(x) = -I_2, \quad \text{i.e.,} \quad  w(\tau({x})) = -{w}^T(x)  
\end{equation}
where $T$ stands for the  transpose of a matrix.
So the time reversal symmetry induces an involution 
$$
\rho: U(2) \rightarrow U(2); \quad g \mapsto -g^{T} \quad s.t. \,\, \rho^2 = 1
$$
 satisfying the commutative diagram
   $$
   \xymatrix{
X \ar[d]^\tau \ar[r]^w &U(2)\ar[d]^\rho\\
X \ar[r]^w          & U(2)}
   $$ 
\end{proof}
  
In other words, $w$ defines a map between involutive spaces    
   \begin{equation} \label{TransFunc}
      w: (X, \tau ) \rightarrow (U(2), \rho)
   \end{equation}
   such that  $\rho \circ w = w \circ \tau$. 
So the time reversal operator $\Theta$ acting on  $\mathcal{H}$ is locally represented by $w \circ K$, where $K$ is the complex conjugation.
In particular,  $w$ is skew-symmetric at any fixed point $x \in X^\tau$, i.e., $ w^T(x) =- w ({x})$.

\begin{examp}
  When $X = \mathbb{S}^3 = \{ (\alpha, \beta) \in \mathbb{C}^2, \, s.t. \, |\alpha|^2 + |\beta|^2 = 1 \}$,
   the time reversal transformation is defined by $\tau (\alpha, \beta) = (\bar{\alpha}, - \beta)$,
  and the fixed points are    $(\alpha, \beta) = (\pm 1, 0)$.
  In addition, the transition function $w$ is given  by
  $$
  w: \mathbb{S}^3 \rightarrow SU(2); \quad ( \alpha, \beta) \mapsto \begin{pmatrix}
                                                                     \beta & \alpha \\
                                                                     -\bar{\alpha} & \bar{\beta}
                                                                    \end{pmatrix}
  $$
\end{examp}

By definition, the Quaternionic K-group $KQ(X, \tau)$ classifies stable isomorphism classes of Quaternionic vector bundles over $X$. 
For the basic examples, the explicit computation of KR-theory is given as follows.
By the decomposition $\mathbb{S}^{1,d} =  \mathbb{R}^{0,d} \cup \{ \infty \} $, one fixed point of TRS is $\infty$ and
the other is $0 \in \mathbb{R}^{0,d}$. Furthermore, one can compute the KR-groups of spheres based on the above decomposition,
$$
KR^{-i} (\mathbb{S}^{1,d}) = KO^{-i}(pt) \oplus KR^{-i}(\mathbb{R}^{0,d}) = KO^{-i}(pt) \oplus  KO^{-i+d}(pt)
$$
Similarly, one can decompose $\mathbb{T}^d = (\mathbb{S}^{1,1})^d$ into fixed points plus involutive Eulidean spaces
$\mathbb{R}^{0, k}$ ($0 \leq k \leq d$) so that the KQ-groups of $\mathbb{T}^d$ are computed based on an iterative decomposition using the same trick,
\begin{equation*}
 KR^{-i}(\mathbb{T}^d) = \oplus_{k =0}^d \binom{d}{k} KO^{-i+k}(pt)
\end{equation*}
The relevant KQ-groups are collected in the following table.

\begin{center}
 \begin{tabular}{||c | c | c | c |c  ||}
 \hline
   & d=1 & d=2 & d=3  \\ [0.5ex]
 \hline\hline
 $\widetilde{KQ}(\mathbb{S}^{1,d})$  & 0 & $\mathbb{Z}_2$ &  $\mathbb{Z}_2$    \\
  \hline
 $\widetilde{KQ}^{-1}(\mathbb{S}^{1,d})$   & $\mathbb{Z}$ & 0 & $\mathbb{Z}_2$    \\
 \hline
  $\widetilde{KQ}(\mathbb{T}^{d})$   & 0 & $\mathbb{Z}_2$ &  $\mathbb{Z}_2^4$    \\
  \hline
 $\widetilde{KQ}^{-1}(\mathbb{T}^{d})$  & $\mathbb{Z}$ & $\mathbb{Z}^2$  & $\mathbb{Z}^3 \oplus \mathbb{Z}_2$    \\ [1ex]
 \hline
\end{tabular}
\end{center}
\begin{rmk}
  In the periodic table, only the strong topological invariants are listed. 
  By index theory, the strong topological invariants of type AII topological insulators belong to $KO^{-2}(pt) = \mathbb{Z}_2$, see below. In the above table, 
  these are the $\mathbb{Z}_2$ components in $\widetilde{KQ}(\mathbb{T}^2)$ and $\widetilde{KQ}^{-1}(\mathbb{T}^3)$, similarly for spheres.
  In addition, the above computation based on the decomposition of torus and sphere gives the same strong topological invariants,
  so   we will focus on spheres and   ignore the weak topological invariants from now on.
\end{rmk}

The geometric object to study in  type AII   topological insulators is Majonara zero modes. 
 Let $H(x)$ be a single-particle Hamiltonian (e.g., a Dirac Hamiltonian) parametrized by the momentum space $ X $, it is time reversal invariant, that is,  for any $x\in X$,  
$\Theta H(x) \Theta^{*} = H(\tau(x))$.  
The effective Hamiltonian of a free fermionic topological insulator is defined by
   \begin{equation*}
      \tilde{H}(x) := \begin{pmatrix}
                     0 & \Theta H(x) \Theta^* \\
                     H(x) & 0
                  \end{pmatrix}
   \end{equation*}  acting on a pair of states $ (\psi, \Theta \psi) \in \Gamma(X, \mathcal{H})$.
Define a real structure  by  $\mathcal{J} := \begin{pmatrix}
                                                      0 & \Theta^* \\
                                                      \Theta & 0
                                                    \end{pmatrix}$ such that  $\mathcal{J}^2 = 1$, 
so the pair $\Psi = (\psi, \Theta \psi)$ is called a Majorana state, since $\Psi$ satisfies the real condition $\mathcal{J} \Psi = \Psi $. 

 \begin{defn}
      A Majorana zero mode  is defined as a
localized Majorana state $\Psi_0 = (\psi_0 , \Theta \psi_0 )$ in a small neighborhood of a fixed point, say $x\in X^\tau$, so that $\psi_0$ and $\Theta \psi_0$
have zero energy at $x$, i.e., $\psi_0(x) = \Theta \psi_0(x) = 0$.
   \end{defn}
     The local geometry of a Majorana zero mode is a conical singularity  $V(x^2 + y^2 - z^2)$ in 3d, 
     that is, $\psi_0$ and  $\Theta \psi_0$ intersect with each other at that fixed point to form a Dirac cone.
 
 \begin{defn}
  The topological $\mathbb{Z}_2$ invariant of a time reversal invariant
topological insulator is defined as the parity of Majorana zero modes.
 \end{defn}

  Near a fixed point, the single-particle Hamiltonian $H$ can be approximated by a Dirac operator $D$, so that the effective Hamiltonian $\tilde{H}$
  is accordingly approximated by a skew-adjoint operator,
  $$
  \tilde{D} := \begin{pmatrix}
               0 & -D \\
               D & 0
              \end{pmatrix}
  $$

  Since $\tilde{H}$ can be approximated by the skew-adjoint operator $\tilde{D}$,
  the topological $\mathbb{Z}_2$ invariant $\nu$ is the  mod 2 analytical index of  $\tilde{H}$, 
\begin{equation}
 \nu = ind_a(\tilde{H}) := \dim \ker \tilde{H}  \quad \text{(mod 2)} 
\end{equation}
which can be computed by the spectral flow of the self-adjoint operator $H$ modulo 2.
In other words, the effective Hamiltonian has KO-dimension $2$, and its analytical index falls into $KO^{-2}(pt)$,
\begin{equation}
     ind_a:  \hat{\mathscr{F}} (L^2(X, \mathcal{H}), \mathcal{J}) \rightarrow  KO^{-2}(pt), \quad ind_a(\tilde{H}) \in \mathbb{Z}_2
\end{equation}
where $\hat{\mathscr{F}}$ denotes the space of skew-adjoint Fredholm operators.

The topological index map in this case is given by
\begin{equation*}\label{topindmap}
   ind_t : KR^{-2}(T^*X) \rightarrow KO^{-2}(pt)  
\end{equation*}
where $\pi: T^*X \rightarrow X$ is the cotangent bundle over $X$.    
Combining it with the Thom isomorphism in $KR$-theory,
\begin{equation*}
   KR^{-j}(X) \cong KR^{d-j}(T^*X)
\end{equation*}
we obtain the topological index map from $KQ$-theory to $KO^{-2}(pt)$.

\begin{examp}[AII 3d]
  When $X = \mathbb{T}^3$, the topological index map is 
   $$
    ind_t: KQ^{-1}(\mathbb{T}^3) = KR^{-5}(\mathbb{T}^3)  \cong KR^{-2}(T^*\mathbb{T}^3)  \rightarrow KO^{-2}(pt) 
   $$
   The transition function of the Hilbert bundle $\pi: \mathcal{H} \rightarrow X$, i.e.,  $w: (X, \tau) \rightarrow (U(2), \rho)$ 
   gives a generator of the top $\mathbb{Z}_2$ component in $\widetilde{KQ}^{-1}(X)$.
   The odd topological index  is the winding number of $w$,
   \begin{equation}\label{wn}
      ind_t(w) = \frac{1}{4\pi^2}\int_X tr(w^{-1}dw)^{3}
   \end{equation}
  which is naturally $\mathbb{Z}_2$-valued due to $w \circ \tau = \rho \circ w$.
\end{examp}

\begin{examp}[AII 2d]\label{2dTopInd}
   When $X = \mathbb{T}^2$,  the topological index map is
   $$
    ind_t: KQ(\mathbb{T}^2) = KR^{-4}(\mathbb{T}^2)  \cong KR^{-2}(T^*\mathbb{T}^2) \rightarrow KO^{-2}(pt) 
   $$
   The integral of the first Chern character modulo 2, i.e., the mod 2 version of the first Chern number ($c_1$ mod 2), cannot give a local
topological index formula in 2d.
 By the $\mathbb{Z}_2$-CW complex structure, $\mathbb{T}^2$ can be decomposed into two copies of $\mathbb{T} \times I$, where $I$ is the unit interval. 
 The Hilbert bundle  $\pi: \mathcal{H} \rightarrow \mathbb{T}^2$ is completely determined by the transition function on the boundary 
 $w: \partial (\mathbb{T} \times I) = S^1 \sqcup S^1 \rightarrow U(2)$. 
 In this case, we propose the topological index formula as 
 \begin{equation}
     ind_t(w) = \frac{1}{4\pi} \int_{S^1 \sqcup S^1} tr(w^{-1}dw)
 \end{equation}
 which is naturally $\mathbb{Z}_2$-valued.
\end{examp}

Based on the Atiyah--Singer index theorem,  the  $\mathbb{Z}_2$ invariant of type AII topological insulators can be understood as a mod 2 index theorem in KR-theory,
that is, the analytical index equals the topological index,
 \begin{equation}
    ind_a(\tilde{H}) = ind_t(w)
 \end{equation}

\subsection{Type AI}
When the time reversal symmetry is even, i.e., $\mathcal{T}^2 = 1$ and denoted by TRS$+1$, 
it falls into the orthogonal class or type AI. 
If $\Theta$ is the time reversal operator representing $\mathcal{T}$, then a Hamiltonian $H$ in this class satisfies
$$
\Theta H(x) \Theta^* = H(\tau(x)), \quad \Theta^2 = 1
$$
Hence the topological band theory of type AI is described by a Real vector bundle $\pi: (\mathcal{H}, \Theta) \rightarrow (X, \tau)$, which is classified by $KR(X)$. 
In this case, KR-theory gives no information in spatial dimensions one, two and three. 
\begin{center}
 \begin{tabular}{|| c | c | c |c  ||}
 \hline
   & d=1 & d=2 & d=3  \\ [0.5ex]
 \hline\hline
 $\widetilde{KR}(\mathbb{S}^{1,d})$ & 0 & 0 &  0  \\
  \hline
\end{tabular}
\end{center}

\section{Particle hole symmetry} \label{sec:PHS}
Topological superconductors gives another important species of general topological insulators,
the relevant discrete symmetry is given by the particle hole symmetry \cite{QZ11, SA17}. 
Before the development of topological insulators, there were several interesting models of topologically protected superconductors in the physics literature. 
For example, Volovik discussed nontrivial topology in 3d superfluid helium 3 ($^3He$) in his book \cite{V03}. 
Read and Green  discussed topologically nontrivial superconductors in a 2d fractional quantum Hall system \cite{RG00},  
and  Kitaev proposed a 1d Majonara chain model \cite{K01}.
The above two models both considered spinless p-wave superconductors, 
which basically describes Majorana zero modes in the vortex core in the 2d case or at the edge in the 1d case. 
After the discovery of topological insulators, topological superconductors are understood in the framework of symmetry protected topological (SPT) phases, 
and later included into the periodic table.

As a remark, in the study of topological superconductors, the $SU(2)$  rotation symmetry and $U(1)$ symmetry (rotation about the $z$-axis in spin space if there is no full $SU(2)$ symmetry) are also very important.
However, KR-theory does not reflect the properties of $U(1)$ or $SU(2)$ symmetry, which could be handled by equivariant twisted K-theory \cite{FM13}. 

A Majorana zero mode in topological superconductors consists of an electron and a hole (also called a Cooper pair), 
which is similar to a Majorana zero mode (or a Dirac cone ) in time reversal invariant topological insulators. 
The quantum vortex in a topological superconductor with odd winding numbers carries an odd number of Majorana zero modes,
which can be understood as an index theorem, for example see \cite{FF10, TDL07}.

Similar to time reversal symmetry, particle hole symmetry (PHS) is the key ingredient in topological superconductors.
In high energy physics, PHS is also called the charge conjugation symmetry, denoted by $\mathcal{C}$.
Based on the mean field method, topological superconductors can be described by effective
Bogoliubov-de Gennes (BdG) Hamiltonians, which  are  invariant under PHS.

\begin{examp}
 In physics, a BdG Hamiltonian is commonly written as
$$
H_{BdG}(\mathbf{k})  = \begin{pmatrix}
            h(\mathbf{k})  &   \Delta(\mathbf{k}) \\
              \Delta^\dagger(\mathbf{k}) &  -h^T(-\mathbf{k})
          \end{pmatrix}
$$  
where $h$ is the normal-state dispersion 
 and $\Delta: X \rightarrow M_2(\mathbb{C})$ is a pairing potential (also called a gap function).
In general, a Cooper pair is formed by pairing two spin-1/2 electrons, the spin angular momentum of a pairing potential is either 0 (spin-singlet) or 1 (spin-triplet).
 For singlet pairings (e.g. d-wave), $\Delta(x) =  \Delta(-x)$; for triplet pairings (e.g. p-wave), $\Delta(x) = - \Delta(-x)$.
A general gap function $\Delta(x) = \Delta_0(x)\sigma_0 + \vec{d}(x) \cdot \vec{\sigma} $ has an even function $\Delta_0$ and an odd vector function
 $\vec{d}$, where  $\vec{\sigma} = (\sigma_x, \sigma_y, \sigma_z)$ is the spin vector. 
\end{examp}

\begin{rmk}
   The dagger symbol $(\dagger)$ is used by physicists to  mean Hermitian conjugate operation, e.g., $A^\dagger = A$ means $A$ is a self-adjoint operator $A^* = A$,
   or a real condition, e.g., $\psi^\dagger = \psi$ means $\psi$ is a real (or Majonara) state. If there are more than one real structures, we have to tell them apart.
   For example, if $C$ is the charge conjugation, then  $\psi^\dagger = \psi$ really means $C\psi = \psi$, and $C$ is different from the adjoint operation $*$.
\end{rmk}

Let $H$ be a self-adjoint Fredholm operator and $C$ be the charge conjugation operator representing PHS. 
If $H$ models an electron, then $CHC^*$ models a hole, which is the anti-particle of an electron. 
The pairing potential operator $\Delta$ describes the pairing between the electron and hole, which takes on different forms depending on the material.
\begin{defn} A BdG Hamiltonian in our notation is defined by
 \begin{equation}
  H_{BdG}(x)  := \begin{pmatrix}
            H(x)  &    C\Delta(x)C^* \\
              \Delta(x) &  CH(x)C^*
          \end{pmatrix} 
 \end{equation}
acting on a quasi-particle $\Psi = (\psi, C\psi)$. 

\end{defn}

The particle hole symmetry defines the same involution on the momentum space as that of time reversal symmetry, i.e.,
  \begin{equation*}
   \tau: X \rightarrow X; \quad   {x} \mapsto -{x}
 \end{equation*}
The particle hole symmetry $\mathcal{C}$ is represented by an anti-unitary operator  called the charge conjugation operator $C$, which can be
written  as a product
$ C :=UK $ where U is a unitary and $K$ is the complex conjugation. A Hamiltonian $H$ is invariant with respect to PHS if and only if
$$
C H(x) C^* = - H(\tau(x)), \quad \forall \,\,x\in X
$$
By assumption, the BdG Hamiltonian is invariant under PHS, i.e.,
$$
C H_{BdG}(x) C^* = - H_{BdG}(\tau(x)), \quad \forall \,\, x \in X
$$
which implies a condition on the pairing potential operator,
$$
C\Delta(x) C^* = - \Delta(\tau(x))
$$

 Let $H$ be a  single-particle Hamiltonian,  $\psi_e$ be an electronic state satisfying the eigenvalue equation
 $$
 H(x)\psi_e(x) = E(x)\psi_e(x), \quad x \in X
 $$
 after applying the charge conjugation operator, one obtains
 $$
 CH(x)C^* \, C\psi_e(x) = CE(x)C\psi_e(x)
 $$
 If $C\psi_e$ is identified with a state $\psi_h$ describing a hole, then the above equation is written as,
 $$
 -H(\tau(x)) \psi_h(\tau(x)) = E(\tau(x))\psi_h(\tau(x))
 $$
 Or equivalently, $\psi_h$  satisfies the following eigenvalue equation,
 $$
 H(x)\psi_h(x) = -E(x)\psi_h(x), \quad x \in X
 $$
 Note that if the electron has eigenvalue $E$, then the hole has eigenvalue $-E$.

 Consider a quasi-particle  $\Psi =  (\psi, C\psi)$, $\Psi$ satisfies the eigenvalue equation
\begin{equation*}\label{EigenEq}
  H_{BdG}(x) \Psi(x) = E_{BdG}(x) \Psi(x), \quad x \in X
\end{equation*}
If the energy scale of the pairing potential $\Delta$ is much smaller than $H$, then $\Delta$ can be neglected and
the above eigenvalue equation  is approximated by
$$
\begin{pmatrix}
  H(x) & 0 \\
  0 & - H(\tau(x)) 
\end{pmatrix}  \begin{pmatrix}
                 \psi(x) \\
                 \psi(\tau(x))
               \end{pmatrix} = \begin{pmatrix}
                                 E(x) & 0 \\
                                 0 & -E(\tau(x))
                               \end{pmatrix} \begin{pmatrix}
                 \psi(x) \\
                 \psi(\tau(x))
               \end{pmatrix}
$$
In this case,  $\Psi$ is called a weak Cooper pair. 

\begin{lem}
  Based on a Bogoliubov transformation, the BdG Hamiltonian $H_{BdG}$ has an off-diagonal form $\tilde{H}_{BdG}$ acting on Bogoliubov quasi-particles.
\end{lem}

\begin{proof}
  Given a quasi-particle $\Psi = (\psi, C\psi)$, one obtains a Bogoliubov quasi-particle, denoted by $\tilde{\Psi}$, by applying the following transformation,
  $$
  \begin{pmatrix}
    i & 1 \\
    1 & -i
  \end{pmatrix} 
  \begin{pmatrix}
    \psi \\
    C \psi
  \end{pmatrix} =  \begin{pmatrix}
                 i \psi +  C\psi \\
                   \psi -i C\psi   
                  \end{pmatrix} = \tilde{\Psi}
  $$
 Define 
 \begin{equation}
   \tilde{H}_{BdG} :=  \begin{pmatrix}
              0  &   H + iC\Delta C^* \\
            CHC^* - i   \Delta &  0
          \end{pmatrix}
 \end{equation}
we have 
 $$
 H_{BdG} \Psi =  \tilde{H}_{BdG} \tilde{\Psi}
 $$
In the physics literature, $\tilde{H}_{BdG}$ is always written diagonally, our motivation is to view it as a chiral system.
More precisely, $C(C\psi + i\psi) = \psi - i C\psi$ and  $C( H - iC\Delta C^*) C^* =  CHC^* + i\Delta$, so this form of  $\tilde{H}_{BdG}$ looks like a Dirac operator $D = \begin{pmatrix}
                                                                                                                              0 & D_+ \\
                                                                                                                              D_- & 0
                                                                                                                             \end{pmatrix}$ with $D_+^* = D_-$ so that $D_+: H_+ \rightarrow H_-$ changes the chirality of spinors.  
Notice that the real structure in our case is given by charge conjugation $C$ instead of adjoint operation $*$.

\end{proof}

\subsection{Type D}
Let PHS$+1$ denote  the even particle hole symmetry such that $\mathcal{C}^2 = 1$,
if a Hamiltonian $H$ is invariant under an even PHS, then the   quantum   system falls into type D.
For instance, in a spinless two-band model, the charge conjugation operator can be defined by $C = \sigma_xK$, here   $\sigma_x$ is the first Pauli matrix with $\sigma_x^2 =1$.

\begin{examp} 
In two dimensions, a spinless chiral  $(p_x \pm ip_y)$-wave    superconductor is a typical example in type D. The BdG Hamiltonian has the form
$$
H_{BdG} = \bar{\Delta} (k_x \tau_x + k_y \tau_y) + h(k_x, k_y) \tau_z = \begin{pmatrix}
                                                                         h(k_x, k_y) & \bar{\Delta} (k_x -ik_y) \\
                                                                         \bar{\Delta} (k_x +ik_y) & - h(k_x, k_y) 
                                                                        \end{pmatrix}
$$
where $\bar{\Delta} \in \mathbb{R}$ is the amplitude of the order parameter, the energy dispersion $h(k_x, k_y) = (k_x^2+ k_y^2)/2m$, 
and the gap function $\Delta(k_+)$ is a linear function in $k_+ = k_x +ik_y$.
\end{examp}

\begin{examp}\label{MajCh}
  The Kitaev Majonara chain is a nontrivial example in 1d, whose effective BdG Hamiltonian is given by 
 $$
 H_{Kitaev} = (-\mu - t\cos (k)) \tau_z + \Delta \sin( k) \tau_y 
 $$
 where $\mu, t, \Delta$ are constant parameters. The nontrivial phase in a Majonara chain is characterized by an unpaired Majonara bound state.
\end{examp}

Define a  real structure $I$ by
$$
I := \begin{pmatrix}
        0 & C \\
        C & 0
     \end{pmatrix}
$$
Since $C^2 = 1$ is even, i.e., $C^* = C$, $I$ is self-adjoint and $I^2 = 1$. 
With the real structure $I$, we have a real condition on $\Psi = (\psi, C \psi)$, i.e.,
$$
I\Psi = \Psi \quad \Longleftrightarrow \quad \begin{pmatrix}
                                           0 & C \\
                                            C & 0
                                           \end{pmatrix}  \begin{pmatrix}
                                                            \psi \\
                                                            C\psi
                                                          \end{pmatrix} =  \begin{pmatrix}
                                                                              \psi \\
                                                                              C\psi
                                                                            \end{pmatrix}
$$
so $\Psi$ is called a Majorana state with respect to $I$.

Notice that the states $\psi_e$ and $\psi_h$ (describing an electron and  a hole)  have different domains,
and they only have a common domain around a neighborhood of a fixed point $x \in X^\tau$, where a vortex can be found. 

\begin{defn}
   A Majorana zero mode in a topological superconductor is defined as a localized Majorana state around a fixed point, 
   denoted by $\Psi_0 = (\psi_0, C\psi_0)$,  such that   $\psi_0(x) = C\psi_0(x) = 0$ for  a fixed point $x \in X^\tau$.
\end{defn}
A Majorana zero mode in a topological superconductor creates a vortex around a fixed point, whose local geometry is the same as a Majorana zero mode in time reversal invariant topological insulators.
The difference is that an electronic state $\psi_e$ has non-negative energy $E_e \geq 0$ and the hole $\psi_h$ has non-positive energy $E_h \leq 0$.  
In a Cooper pair $(\psi, C \psi)$, one can tell them apart by looking at the sign of energy, but a Kramers pair $(\psi, \Theta \psi)$ have the same energy.  
In a Majonara zero mode $ (\psi_0, C\psi_0)$,  $\psi_0$ and $C\psi_0$ touch with each other at a fixed point $x\in X^\tau$ with $E(x) = 0$.   


Since the energy of a hole is minus that of the corresponding electron, a pair $(\psi, C\psi)$ has energy levels $(E_n, E_{-n}= -E_n)$.
Majonara zero modes are the geometric objects to study, we only consider $\Psi_0 = (\psi_0, C\psi_0)$ with intersections at the zero energy. 
If a Hilbert bundle $\pi: (\mathcal{H}, C) \rightarrow (X, \tau)$  is defined to model Majonara zero modes of a topological superconductor, 
then $\mathcal{H}$ is assumed to have rank 2. 
Now the Hilbert bundle of a topological superconductor can be compared to that of a topological insulator.
Indeed, the Hilbert bundle $\pi: (\mathcal{H}, C) \rightarrow (X, \tau)$ with $\tau^2 =1$ and $C^2 = 1$  is a Real vector bundle.

\begin{lem} The transition function $w: X \rightarrow U(2)$ representing PHS+1 has the property
  \begin{equation}
        w(\tau(x)) = w^T(x)   
  \end{equation}
  which induces an involution $\rho$ on $U(2)$ s.t. $w \circ \tau = \rho \circ w$.
\end{lem}

\begin{proof} The proof is similar to that in Lemma \ref{LemTRS}.
  Due to $C^2 = 1$, we have 
  $$
   w(\tau(x))\overline{w(x)} = 1 \quad \Leftrightarrow \quad  w(\tau(x)) = w^T(x)   
  $$
  The involution $\rho$ is defined by
  $$
  \rho: U(2) \rightarrow U(2); \quad g \mapsto g^T
  $$
  so that the above property is equivalent to the commutative diagram
   $$
   \xymatrix{
X \ar[d]^\tau \ar[r]^w &U(2)\ar[d]^\rho\\
X \ar[r]^w          & U(2)}
   $$ 
\end{proof}

 \begin{rmk}
   If the particle hole symmetry is even, the transition function is a symmetric matrix at the fixed points
   $$
    w(x) = w^T(x), \quad x \in X^\tau
   $$
   which is different from the case in TRS$-1$.
   
 \end{rmk}

Type D topological superconductors can be classified by the KR-group $KR^{-2}(X)$. 
\begin{center}
 \begin{tabular}{||c | c | c | c |c  ||}
 \hline
   &  d=1 & d=2 & d=3 \\ [0.5ex]
 \hline\hline
 $\widetilde{KR}^{-2}(\mathbb{S}^{1,d})$ &  $\mathbb{Z}_2$ & $\mathbb{Z}$ &  0    \\ [1ex]
 \hline
\end{tabular}
\end{center}
Notice the difference between the Hilbert bundles with symmetries TRS+1 and PHS+1. 
Recall that a Hamiltonian $H$ is invariant under TRS+1 if $\Theta H(x) \Theta^* = H(\tau(x))$, and the relevant Hilbert bundle is
$\pi: (\mathcal{H}, \Theta) \rightarrow (X, \tau)$ with $\tau^2 =1$ and $\Theta^2 = 1$. However, a type D Hamiltonian satisfies
$C H(x) C^* = -H(\tau(x))$, the negative sign is the key difference. As a Clifford bundle, 
$(\mathcal{H}, \Theta)$ has an action by the real Clifford algebra $C\ell_{0,0} = \mathbb{R}$, 
and $(\mathcal{H}, C)$ has an action by $C\ell_{0,2} = \mathbb{H}$. As a consequence, 
$(\mathcal{H}, \Theta) \rightarrow (X, \tau)$ is classified by $KR(X)$, and 
$(\mathcal{H}, C) \rightarrow (X, \tau)$ is classified by $KR^{-2}(X)$,
which can be also seen from the effective BdG Hamiltonian and the relevant index.


 If the effective BdG Hamiltonian is given by 
 $$
 \tilde{H}_{BdG} = \begin{pmatrix}
                    0 & H + i C \Delta C^* \\
                 CHC^* -i   \Delta & 0
                   \end{pmatrix}  
 $$
 which is a Fredholm operator by assumption,
 then the Fredholm index of $\tilde{H}_{BdG}$  is defined as usual, 
  \begin{equation}
             ind_a(\tilde{H}_{BdG} )  = \dim ker (H + i C \Delta C^*) - \dim ker (CH C^*-i \Delta)
  \end{equation}
  which counts the net change of zero modes of Bogoliubov quasi-particles.
  
  The analytical index of a BdG Hamiltonian is essentially determined by  the pairing potential $\Delta$. 
  In particular, if $ \Delta  \sim D$ can be approximated by a  Dirac-type operator and $H \sim 0$ is ignored as a higher order perturbation, see \cite{FF10, TDL07} for examples,
 then the above formula gives zero,
 i.e., $ ind_a(\tilde{H}_{BdG} ) = 0 $, which is of course not a good topological invariant.  
 However, in this case
 \begin{equation}
  \tilde{H}_{BdG} \sim \hat{H}_{BdG} = \begin{pmatrix}
                              0 & i C D C^*  \\
                              -iD & 0
                             \end{pmatrix} 
 \end{equation}
 is analogous to a  skew-adjoint ($A^* = -A$) Fredholm operator   
 \begin{equation}
   C \hat{H}_{BdG} C^* = - \hat{H}_{BdG}
 \end{equation}
 Similar to the mod 2 analytical index of a skew-adjoint Fredholm operator, we define the mod 2 index of $\hat{H}_{BdG}$ as
 \begin{equation}
   ind_a(\hat{H}_{BdG}) := \dim ker(\hat{H}_{BdG}) \quad \text{(mod 2)} \,\, \in KO^{-2}(pt)
 \end{equation} 
 which can be reduced to the parity of zero modes of the Dirac operator
 $$
    ind_a(\hat{H}_{BdG}) := \dim ker(D) \quad \text{(mod 2)} 
 $$
 This mod 2 analytical index counts the parity of Majonara zero modes, and 
 the topological $\mathbb{Z}_2$ invariant of a topological superconductor falls into $KO^{-2}(pt) = \mathbb{Z}_2$.

The set of fixed points  supports  Majorana zero modes, and the parity of Majorana zero modes can be interpreted as a mod 2 analytical index.
Furthermore, a topological index can be used to compute the analytical index.

\begin{examp}[D 2d]
 In two dimensions, the topological index map is 
 $$
 ind_t: KR^{-2}(X) = KR(TX) \rightarrow KO(pt) 
 $$
 In addition, the topological index can be computed by the first Chern number
 $$
 ind_t(p) = \frac{1}{2\pi} \int_X tr(pdpdp)
 $$
 where $p$ is a projection representing the Hilbert bundle $\pi: (\mathcal{H}, C) \rightarrow (X, \tau)$.
\end{examp}

\begin{examp}[D 1d]
 In one dimensions, the topological index map is an odd index 
 $$
  ind_t: KR^{-3}(X) = KR^{-2}(TX) \rightarrow KO^{-2}(pt) 
 $$
  The transition function $w: (X, \tau) \rightarrow (U(2), \rho)$ gives a class $[w] \in KR^{-3}(X)$, and 
  $$
  ind_t(w) = \frac{1}{2 \pi} \int_X tr(w^{-1}dw)
  $$
  If $KR^{-2}(X)$ were used to classify the relevant Real bundles over the one-dimensional $X$, it is possible to obtain the $\mathbb{Z}_2$-invariant in $KO^{-1}(pt)$. 
  However, the analysis of the effective BdG Hamiltonian $\hat{H}_{BdG}$ tells us that the $\mathbb{Z}_2$ index belongs to $KO^{-2}(pt)$, which forces the topological index maps from 
  $KR^{-3}(X)$ to $KO^{-2}(pt)$. In addition, on an odd-dimensional manifold, an odd topological index is always expected as the bulk theory. The shift from $KR^{-2}(X)$ to $KR^{-3}(X)$
  has an intimate relation to spectral flow, which is a  source of the odd index theorem of connections (or gauge transformations).
\end{examp}

 \begin{lem}
     The above  odd index $ind_t(w)$ is naturally $\mathbb{Z}_2$-valued.
 \end{lem}
\begin{proof} Applying $\tau$ on the base manifold $X$ changes local coordinates, 
    $$
   \frac{1}{2 \pi} \int_X tr(w^{-1}(\tau(x))d(\tau(x))w(\tau(x)))
  $$
  The exterior derivative changes sign $d(\tau(x)) = d(-x) = -d(x)$, and the compatibility condition $w(\tau(x)) = w^T(x)$ is used to rewrite the above as
     $$
   - \frac{1}{2 \pi} \int_X tr((w^T)^{-1}dw^T)
  $$
  which is the minus  winding number of $w^T$. Due to the canonical relation $tr \ln M = \ln \det M $ for a matrix $M$ and $\det M = \det M^T$, the winding number of 
  $w$ equals that of $w^T$. As a global invariant, the winding number does not depend on the choice of local coordinates, so for the involutive space $(X, \tau)$, 
  the topological index $ind_t(w)$ should be identified with $-ind_t(w)$ from the above computation. In other words, $ind_t(w)$ must be $\mathbb{Z}_2$-valued due to 
  the real structure $\tau$.
\end{proof}

   \subsection{Type C}
   If the particle hole symmetry is odd, i.e., $\mathcal{C}^2 = -1$, denoted by PHS$-1$, then the system 
   falls into type C.  
   For example, the charge conjugation operator can be defined as $ C = -i\tau_y K = \begin{pmatrix}
                                                                                      0 & -1 \\
                                                                                      1 & 0
                                                                                     \end{pmatrix} K $ 
   
   \begin{examp}
    A spin-singlet  $(d + id)$-wave  superconductor in two dimensions has the BdG Hamiltonian
    $$
    H(\textbf{k}) = \Delta_{x^2 -y^2}(k_x^2-k_y^2)\tau_x +\Delta_{xy}k_xk_y\tau_y  + \frac{(k_x^2 + k_y^2)}{2m}\tau_z 
    $$
    or in a matrix form,
    $$
    H_{BdG} = \begin{pmatrix}
                          (k_x^2 + k_y^2)/2m &  \Delta_{x^2 -y^2}(k_x^2-k_y^2) - i \Delta_{xy}k_xk_y   \\
                          \Delta_{x^2 -y^2}(k_x^2-k_y^2) + i \Delta_{xy}k_xk_y  & -(k_x^2 + k_y^2)/2m
              \end{pmatrix}
    $$
   \end{examp}

    The Hilbert bundle in this class is defined by $\pi: (\mathcal{H}, C) \rightarrow (X, \tau)$ such that $\tau^2 = 1$ and  $C^2 = -1$.
    Since the charge conjugation operator is odd,  we  view the Hilbert bundle as a Quaternionic  bundle and
     the band structure is classified by KQ-theory.
    By the same reason as in type D, there exists a shift by $2-$  (due to $C H(x) C^* = -H(\tau(x))$) and the relevant KR-group is 
   $$
   KR^{-6}(X)= KQ^{-2}(X)
   $$
  
   \begin{center}
 \begin{tabular}{||c | c | c | c |c  ||}
 \hline
   &   d=1 & d=2 & d=3  \\ [0.5ex]
 \hline\hline
 $\widetilde{KR}^{-6}(\mathbb{S}^{1,d})$ &   0 & $\mathbb{Z}$ &  0    \\
  \hline
\end{tabular}
\end{center} 

If the effective BdG Hamiltonian is given by 
 $$
 \tilde{H}_{BdG} = \begin{pmatrix}
                    0 & H + i C \Delta C^* \\
                 CHC^* -i   \Delta & 0
                   \end{pmatrix}  $$
 then the quaternionic Fredholm index of $\tilde{H}_{BdG}$  is  
          $$
          ind_a(\tilde{H}_{BdG} )  = \dim ker_{\mathbb{H}} (H + i C \Delta C^* ) - \dim ker_{\mathbb{H}} (CHC^* -i \Delta)
          $$
 Because of the Quaternionic structure $C$ s.t. $C^2 = -1$, each Bogoliubov quasi-particle $i\psi + C\psi$  (or its conjugate under $C$) gives a quaternionic vector when evaluated at a point $x\in X$.

\begin{examp}[C 2d]
 In two dimensions, the topological index map is 
 $$
 ind_t: KR^{-6}(X) = KR^{-4}(TX) \rightarrow KO^{-4}(pt) 
 $$
 In addition, the topological index can be computed by the first Chern number
 $$
 ind(p) = \frac{1}{2\pi} \int_X tr(pdpdp)
 $$
 where $p$ is a projection representing the Hilbert bundle $\pi: (\mathcal{H}, C) \rightarrow (X, \tau)$.
\end{examp}

\subsection{Type DIII}

When a topological quantum system has both odd time reversal symmetry and even particle hole symmetry, i.e., TRS$-1$ `+' PHS$+1$, it falls into type DIII. 
In this class, the constraints on a Hamiltonian are
$$
\Theta H(x) \Theta^* = H(\tau(x)), \quad CH(x)C^* = -H(\tau(x)),  \quad \Theta^2 = -1, \,\, C^2 = 1
$$
For example, the time reversal operator is defined by $\Theta = \tau_0 \otimes i\sigma_y K$ and the charge conjugation is defined by $C =\tau_x \otimes \sigma_0K$,
$\sigma_i$ and $\tau_i$ are Pauli matrices operating in spin space and particle-hole space.
In general, $\Theta$ and $C$ are commutative, i.e., $\Theta C = C \Theta$.
For example,  the three-dimensional superfluid $^3He$ phase B  is in this class, another example is
given by the superposition of $(p + ip)$- and $(p - ip)$-wave superconductors in two dimensions. 

\begin{examp}\label{SPpwave}
   An equal superposition of two chiral $p$-wave superconductors with opposite chiralityies has the BdG Hamiltonian
   $$
   H_{BdG}(\mathbf{k}) = \begin{pmatrix}
            h(\mathbf{k})  &   \Delta(\mathbf{k}) \\
              \Delta^\dagger(\mathbf{k}) &  -h^T(-\mathbf{k})
             \end{pmatrix}
   $$
   where $h(\mathbf{k}) = \varepsilon(\mathbf{k}) \tau_0$ with $\varepsilon(\mathbf{k})$ the energy dispersion of a single particle, and the pairing potential $\Delta$ is given by 
   $$
    \Delta(\mathbf{k}) = \bar{\Delta} \begin{pmatrix}
                                       k_x + ik_y & 0 \\
                                       0 & -k_x + ik_y
                                      \end{pmatrix}
   $$
\end{examp}

If one considers the product of TRS$-1$ and PHS$+1$, which is the chiral symmetry (CS) and denoted  by $S = \Theta C$,  
then $S$ is a unitary operator such that $S^2 = -1$. For example, $S =  \tau_x \otimes i \sigma_y$. 
So the constraint with respect to the chiral symmetry is
$$
S H(x) S^* = - H(x), \quad S^2 = -1
$$
By convention, the system is equivalently described by 
\begin{equation}
  CH(x)C^* = -H(\tau(x)),  \quad S H(x) S^* = - H(x), \quad C^2 = 1, \,\, S^2 = -1
\end{equation}

In the presence of the chiral symmetry, the effective BdG Hamiltonian in type DIII  can be written off-diagonally 
$$
\tilde{H}_{BdG} = \begin{pmatrix}
          0 & Q \\
          Q^* & 0
       \end{pmatrix}  \quad \text{with} \quad CQ(x)C^* = -Q(\tau(x))
$$
which is a self-adjoint Fredholm operator. The Fredholm index is defined as
   $$
   ind_a(\tilde{H}_{BdG} )  = \dim ker (Q) - \dim ker (Q^*)
   $$
   
The Hilbert bundle $\mathcal{H}$ is now equipped with one involution and one anti-involution, i.e., $C^2 = 1$ and $S^2 = -1$
so that each band is quadruply degenerate. If each fiber is viewed as a vector space over $\mathbb{H} \oplus \mathbb{H}$ due to $C$ and $S$, 
then the Hilbert bundle has an action by $C\ell_{0, 3} = \mathbb{H} \oplus \mathbb{H}$. Compared to type D (classified by $KR^{-2}(X)$), 
the Hilbert bundles in type DIII can be classified by $KR^{-3}(X)$. 

\begin{center}
 \begin{tabular}{||c | c | c | c |c   ||}
 \hline
   &  d=1 & d=2 & d=3   \\ [0.5ex]
 \hline\hline
 $\widetilde{KR}^{-3}(\mathbb{S}^{1,d})$  &  $\mathbb{Z}_2$ & $\mathbb{Z}_2$ & $\mathbb{Z}$    \\
 \hline
\end{tabular}
\end{center}

By the form of the above Fredholm index, the topological index is a map
$$
ind_t : KR(TX) \rightarrow KO(pt)
$$
Combine it with the Thom isomorphism in the 3d case, it can be written as
$$
ind_t: KR^{-3}(X) \rightarrow KO(pt)
$$

\begin{examp}[DIII 3d]

   In three dimensions, the topological index can be computed by the winding number of the transition function $w : X \rightarrow U(4)$,
   $$
   ind_t(w) = \frac{1}{4\pi^2} \int_X tr(w^{-1}dw)^3
   $$
   A typical example is given by the B phase of superfluid $^3He$ \cite{V03}.
\end{examp}

\begin{examp}[DIII 2d]
   In the two-dimensional case, if $Q$ can be approximated by a Dirac operator and accordingly $\tilde{H}_{BdG}$ is approximated by a skew-adjoint Fredholm operator, see Ex. \ref{SPpwave},
   then $\tilde{H}_{BdG}$ has a mod 2 analytical index,
   $$
   ind_a(\tilde{H}_{BdG}) = \dim ker (Q) \quad \text{(mod 2)} 
   $$
   which counts the parity of Majonara zero modes in 2d topological superconductors.
   
   The relevant topological index can be treated as  in type AII, see Ex. \ref{2dTopInd}.
   Let us use $\mathbb{T}^2$  as the momentum space,    the topological index map is given by 
   $$
   ind_t: {KR}^{-4}(\mathbb{T}^{2})  = KR^{-2}(T\mathbb{T}^{2}) \rightarrow KO^{-2}(pt)
   $$ 
   By the $\mathbb{Z}_2$-equivariant CW-complex structure of $\mathbb{T}^2$ induced by TRS or PHS, we decompose $\mathbb{T}^2$ into $\mathbb{T} \times I$ with a free $\mathbb{Z}_2$ action. 
   The Hilbert bundle $\pi: \mathcal{H} \rightarrow \mathbb{T}^2$ is completely determined by the transition function $w : S^1 \sqcup S^1 \rightarrow U(4)$, which gives rise to a class
   $[w] \in KR^{-3}( S^1 \sqcup S^1 )$. We propose the topological index formula as
   \begin{equation}
     ind_t(w) = \frac{1}{4\pi} \int_{ S^1 \sqcup S^1 } tr(w^{-1} dw)
  \end{equation}

\end{examp}

\begin{examp}[DIII 1d]
   The one-dimensional type DIII case is a variant of the Kitaev Majonara chain in Ex. \ref{MajCh}, which also has a $\mathbb{Z}_2$ invariant \cite{BA13}.
   The topological index in this case  can be computed as a 1d winding number
   $$
   ind_t(w) = \frac{1}{2\pi} \int_X tr(w^{-1}dw)
   $$
\end{examp}

\subsection{Type CI}
When a topological quantum system has both even time reversal symmetry and odd particle hole symmetry, i.e., TRS$+1$ `+' PHS$-1$, it falls into type CI. In this class, the constraints on a Hamiltonian are
$$
\Theta H(x) \Theta^* = H(\tau(x)), \quad CH(x)C^* = -H(\tau(x)),  \quad \Theta^2 = 1, \,\, C^2 = -1
$$
 A 2d d-wave ($d_{x^2 -y^2}$-wave) spin-singlet superconductor gives an example of type CI topological superconductors with TRS.
The chiral symmetry is   the product  $S = \Theta C$,  and $S$ is a unitary operator such that $S^2 = -1$. So the constraint with respect to the chiral symmetry is
$$
S H(x) S^* = - H(x), \quad S^2 = -1
$$
By convention, the system is equivalently described by 
\begin{equation}
  CH(x)C^* = -H(\tau(x)),  \quad S H(x) S^* = - H(x), \quad C^2 = -1, \,\, S^2 = -1
\end{equation}


In the presence of the chiral symmetry, the effective BdG Hamiltonian in type CI  can be written off-diagonally 
$$
\tilde{H}_{BdG} = \begin{pmatrix}
          0 & Q \\
          Q^* & 0
       \end{pmatrix}, \quad \text{with} \quad CQ(x)C^* = Q(\tau(x))
$$
which is a self-adjoint Fredholm operator. The Fredholm index is defined as
   $$
   ind_a(\tilde{H}_{BdG} )  = \dim ker_{\mathbb{H}} (Q) - \dim ker_{\mathbb{H}} (Q^*)
   $$
where $Q$ is viewed as a quaternionic operator due to  $C^2 = -1$.

From type DIII to CI, one switches the sign of  PHS from $+1$ to $-1$, so the topological index is a map
$$
ind_t : KQ(TX) = KR^{-4}(TX) \rightarrow KO^{-4}(pt)
$$
Combining it with the Thom isomorphism, for the three-dimensional case, 
$$
ind_t: KR^{-7}(X) \rightarrow KO^{-4}(pt)
$$

\begin{center}
 \begin{tabular}{||c | c | c | c |c  ||}
 \hline
   &  d=1 & d=2 & d=3  \\ [0.5ex]
 \hline\hline
 $\widetilde{KR}^{-7}(\mathbb{S}^{1,d})$  &  0 & 0 & $\mathbb{Z}$   \\
  \hline
\end{tabular}
\end{center}

\begin{examp}[CI 3d]
   In three dimensions, the topological index can be computed by the winding number of the transition function $w : X \rightarrow U(4)$,
   $$
   ind_t(w) = \frac{1}{4\pi^2} \int_X tr(w^{-1}dw)^3
   $$
\end{examp}

\section{Chiral symmetry}\label{sec:CS}

The chiral symmetry (CS) or sublattice symmetry (SLS), denoted by $\mathcal{S}$,  can be defined as the product of TRS $\mathcal{T}$ and PHS $\mathcal{C}$,
$$
\mathcal{S} = \mathcal{T} \cdot \mathcal{C}
$$ which is represented by a unitary operator $S$. If both TRS   and PHS are present, there always exists a chiral symmetry.
 When both $\mathcal{T}$ and $\mathcal{C}$ are absent, it is still possible to have a chiral symmetry, which is the chiral unitary class or type AIII. 
 In this paper, we do not consider complex K-theory, which can be used to classify type A and AIII topological insulators.
 A Hamiltonian $H$ is invariant under the chiral symmetry if and only if
$$
 S  H(x) S^* = - H(x), \quad x \in X
$$

\begin{examp}
The Hamiltonian of a system with sublattice symmetry is generally defined as
$$
H = \begin{pmatrix}
      0 & T \\
      T^* & 0
    \end{pmatrix}
$$
For example, in a bipartite lattice, $T$ collects the hopping amplitudes from one sublattice to another.
In this case, the sublattice symmetry $\mathcal{S}$ can be represented by
$$
S = \begin{pmatrix}
                I_n & 0 \\
                0 & -I_m
              \end{pmatrix}
$$
where $I_n$ is the identity matrix of size $n$, here $n$ and $m$ are the numbers of atoms in sublattices A and B.
\end{examp}

\subsection{Type CII}
When a topological quantum system has both odd time reversal symmetry and odd particle hole symmetry,
i.e., TRS$-1$ `+' PHS$-1$, it falls into the chiral symplectic class or type CII. A Hamiltonian $H$ is in type CII if and only if
$$
\Theta H(x) \Theta^* = H(\tau(x)), \quad CH(x)C^* = -H(\tau(x)),  \quad \Theta^2 = -1, \,\, C^2 = -1
$$
The chiral symmetry is the product $S = \Theta C$ and in this case $S^2 = 1$. 
So the constraint with respect to the chiral symmetry is  
$$
S H(x) S^* = - H(x), \quad S^2 = 1
$$
The system  is equivalently described by 
\begin{equation}
  \Theta H(x) \Theta^* = H(\tau(x)),  \quad S H(x) S^* = - H(x), \quad \Theta^2 = -1, \,\, S^2 = 1
\end{equation}

In the presence of the chiral symmetry, the effective  Hamiltonian   can be written off-diagonally 
$$
{H} = \begin{pmatrix}
          0 & Q \\
          Q^* & 0
       \end{pmatrix}  
$$
For example, $Q$ is a Dirac Hamiltonian in 3d.
The Fredholm index is 
   $$
   ind_a({H} )  = \dim ker_{\mathbb{H}} (Q) - \dim ker_{\mathbb{H}} (Q^*)
   $$
where $Q$ is viewed as a quaternionic operator due to  $\Theta^2 = -1$.

The Hilbert bundle $\mathcal{H}= \mathcal{H}_+ \oplus \mathcal{H}_-$ is $\mathbb{Z}_2$-graded due to the chiral symmetry, and
 $\mathcal{H}_\pm$ is a Quaternionic vector bundle by TRS$-1$.  
In this class, the topological index is a map
$$
ind_t : KQ(TX) = KR^{-4}(TX) \rightarrow KO^{-4}(pt)
$$
Combining it with the Thom isomorphism, for the one-dimensional case, 
$$
ind_t: KR^{-5}(X) \rightarrow KO^{-4}(pt)
$$

\begin{center}
 \begin{tabular}{||c   | c | c |c  ||}
 \hline
   &   d=1 & d=2 & d=3  \\ [0.5ex]
 \hline\hline
 $\widetilde{KR}^{-5}(\mathbb{S}^{1,d})$    & $\mathbb{Z}$ & 0 & $\mathbb{Z}_2$    \\
  \hline
\end{tabular}
\end{center}

\begin{examp}[CII 1d]
   In one dimensions, the topological index can be computed by the winding number of the transition function $w : X \rightarrow U(4)$
   $$
   ind_t(w) = \frac{1}{2\pi} \int_X tr(w^{-1}dw)
   $$
\end{examp}

\begin{examp}[CII 3d]
In three dimensions,  the situation is similar to the case of $\mathbb{Z}_2$-topological insulators in class AII.  
In this case, the effective Hamiltonian takes the form 
$$
H = \begin{pmatrix}
      0 & D \\
      D & 0
    \end{pmatrix}
$$
where $D$ is a self-adjoint Dirac-like  operator.
As in type AII, we define the analytical index by
$$
ind_a(H) = \dim ker (D) \quad \text{(mod 2)}
$$
So the topological index map is
$$
ind_t: KR^{-5}(X) \rightarrow KO^{-2}(pt)
$$
and a local formula is given by a 3d winding number
$$
ind_t(w) = \frac{1}{4\pi^2} \int_X tr(w^{-1}dw)^3
$$
which is $\mathbb{Z}_2$-valued.
\end{examp}

\subsection{Type BDI}

When a topological quantum system has both even time reversal symmetry and even particle hole symmetry,
i.e., TRS$+1$ `+' PHS$+1$, it falls into the chiral orthogonal class or type BDI. A Hamiltonian $H$ is in type BDI if and only if
$$
\Theta H(x) \Theta^* = H(\tau(x)), \quad CH(x)C^* = -H(\tau(x)),  \quad \Theta^2 = 1, \,\, C^2 = 1
$$
The chiral symmetry is the product $S = \Theta C$ with $S^2 = 1$. So the constraint in terms of $S$ is  
$$
S H(x) S^* = - H(x), \quad S^2 = 1
$$ 
The system  is equivalently described by 
\begin{equation}
  \Theta H(x) \Theta^* = H(\tau(x)),  \quad S H(x) S^* = - H(x), \quad \Theta^2 = 1, \,\, S^2 = 1
\end{equation}

In the presence of the chiral symmetry, the effective  Hamiltonian   can be written off-diagonally 
$$
{H} = \begin{pmatrix}
          0 & Q \\
          Q^* & 0
       \end{pmatrix}  
$$
The Fredholm index is 
   $$
   ind_a({H} )  = \dim ker (Q) - \dim ker (Q^*)
   $$

The Hilbert bundle $\mathcal{H}= \mathcal{H}_+ \oplus \mathcal{H}_-$ is $\mathbb{Z}_2$-graded due to the chiral symmetry, and
 $\mathcal{H}_\pm$ is a Real vector bundle by TRS$+1$.  
In this class, the topological index is a map
$$
ind_t : KR(TX)   \rightarrow   KO (pt)
$$
Combining it with the Thom isomorphism, for the one-dimensional case, 
$$
ind_t: KR^{-1}(X) \rightarrow KO(pt)
$$

\begin{center}
 \begin{tabular}{|| c | c | c |c ||}
 \hline
    & d=1 & d=2 & d=3  \\ [0.5ex]
 \hline\hline
 $\widetilde{KR}^{-1}(\mathbb{S}^{1,d})$   & $\mathbb{Z}$ & 0 & 0   \\
  \hline
\end{tabular}
\end{center}

\begin{examp}[BDI 1d]
   In one dimensions, the topological index can be computed by a 1d winding number  
   $$
   ind_t(w) = \frac{1}{2\pi} \int_X tr(w^{-1}dw)
   $$
\end{examp}

\section{Main results} \label{main}

Now we collect the $\mathbb{Z}$ and $\mathbb{Z}_2$ invariants of topological insulators and superconductors together, 
and write them in terms of KO-theory of a point, for simplicity $KO^{-i} = KO^{-i}(pt)$. Notice that for each $\mathbb{Z}_2$ invariant, 
we use $KO^{-2}$, which is physically determined by the evenness of quasi-particles. More precisely, a $\mathbb{Z}_2$ invariant basically counts the parity of Majorana zero modes,
which consists of two real particles. 

\begin{center}
 \begin{tabular}{||c|  c | c | c | c |c |c ||}
 \hline
    Type  & d=1 & d=2 & d=3  \\ [0.5ex]
 \hline\hline
  AI     & 0 &  0  & 0 \\
  \hline
 BDI      & $  KO $& 0  & 0 \\
 \hline
 D    & $KO^{-2}$ &  $  KO $  & 0 \\
  \hline
 DIII     & $ KO^{-2}$ & $KO^{-2}$ &   $  KO$  \\  
 \hline
  AII    & 0 & $  KO^{-2}$ &  $KO^{-2}$   \\
  \hline
CII   & $  KO^{-4}$ & 0 & $  KO^{-2}$ \\
 \hline
 C     & 0 & $  KO^{-4}$ &  0 \\
  \hline
 CI    & 0  & 0  & $  KO^{-4}$ \\ [1ex]
 \hline
\end{tabular}
\end{center}

\begin{thm}\label{idmap}
  The diagonal map is the identity map induced by the $(1,1)$-periodicity of the KR-theory,
  $$
  KR^{p-q} = KR^{p, q} \cong KR^{p+1, q+1}
  $$
\end{thm}

\begin{cor}
  Every $\mathbb{Z}_2$ invariant of topological insulators and superconductors belongs to $KO^{-2}(pt)$.
\end{cor}
 \begin{proof}
    The topological $\mathbb{Z}_2$ invariant in type AII is used to count the parity of Majonara zero modes, which is a quasi-particle described by a complex skew-adjoint Fredholm operator.
    As a result, the analytical index maps the effective Hamiltonian to $KO^{-2}(pt)$, which is always even. Combining with the identity diagonal map, one completes the proof. 
 \end{proof}

First of all, the topological invariants are completely determined by the effective Hamiltonians of topological insulators and superconductors. 
The discrete symmetries such the TRS, PHS and CS place symmetry constraints on the effective Hamiltonians, which fall into different subclasses of Fredholm operators.
On the other hand, those discrete symmetries introduce real structures and involutions in the topological band theory, which falls into the category of KR-theory. 
The classifying spaces of KR-theory can be described by subclasses of Fredholm operators as a generaliztion of the Atiyah--J{\"a}nich theorem \cite{AS69, GS15, L88}. 
Let $ \mathscr{F}_i$ be a subspace of Fredholm operators as a classifying space of $KR^{-i}$, that is,  for a compact space $X$,
$$
KR^{-i}(X) = [X, \mathscr{F}_i]
$$
The analytical index of the  effective Hamiltonians $H$ gives the topological invariants, i.e., if $H \in \mathscr{F}_i$, 
$$
ind_a(H) \in \pi_0(\mathscr{F}_i) = KO^{-i}(pt), \quad i = 0, 1, 2, 4 
$$
where the KO-theory of a point can be identified with the homotopy theory of Fredholm operators. Now we collect the types of the effective Hamiltonians into a table.

\begin{center}
 \begin{tabular}{||c|  c | c | c | c |c |c ||}
 \hline
    Type  & d=1 & d=2 & d=3  \\ [0.5ex]
 \hline\hline
  AI     &  &    &  \\
  \hline
 BDI      & $\mathscr{F}_0$ &   &  \\
 \hline
 D    &  $\mathscr{F}_1$ &   $\mathscr{F}_0$   &  \\
  \hline
 DIII     & $\mathscr{F}_2$ & $\mathscr{F}_1$ &    $\mathscr{F}_0$  \\  
 \hline
  AII    &  & $\mathscr{F}_2$ &  $\mathscr{F}_1$  \\
  \hline
CII   & $\mathscr{F}_4$ &  & $\mathscr{F}_2$ \\
 \hline
 C     &  & $\mathscr{F}_4$ &   \\
  \hline
 CI    &   &   & $\mathscr{F}_4$ \\ [1ex]
 \hline
\end{tabular}
\end{center}

In order to compute the topology invariants from the effective Hamiltonians, one can calculate the first two homotopy groups $\pi_0$ and $\pi_1$.  
The homotopy group $\pi_0$ counts the connected components and $\pi_1$ is related to spectral flow. 
Notice that for the subspace $\mathscr{F}_1$, we compute its $\pi_1$ group to get an invariant in $\pi_1(\mathscr{F}_1) = \pi_0(\mathscr{F}_2) = KO^{-2}(pt)$.

\begin{center}
 \begin{tabular}{||c|  c | c | c | c |c |c ||}
 \hline
    Type  & d=1 & d=2 & d=3  \\ [0.5ex]
 \hline\hline
  AI     &  &    &  \\
  \hline
 BDI      & $\pi_0(\mathscr{F}_0)$ &   &  \\
 \hline
 D    &  $\pi_1(\mathscr{F}_1)$ &   $\pi_0(\mathscr{F}_0)$   &  \\
  \hline
 DIII     & $\pi_0(\mathscr{F}_2)$ & $\pi_1(\mathscr{F}_1)$ &    $\pi_0(\mathscr{F}_0)$  \\  
 \hline
  AII    &  & $\pi_0(\mathscr{F}_2)$ &  $\pi_1(\mathscr{F}_1)$  \\
  \hline
CII   & $\pi_0(\mathscr{F}_4)$ &  & $\pi_0(\mathscr{F}_2)$ \\
 \hline
 C     &  & $\pi_0(\mathscr{F}_4)$ &   \\
  \hline
 CI    &   &   & $\pi_0(\mathscr{F}_4)$ \\ [1ex]
 \hline
\end{tabular}
\end{center}

By the Atiyah--Singer index theorem, the analytical index can be computed by the topological index. Recall the topological index map
$$
ind_t: KR^{-i}(TX) \rightarrow KO^{-i}(pt)
$$
(the symbol class of) an effective Hamiltonian determines a KR-theoretic class of the  cotangent bundle.
The following table gives the sources of the topological index map.
\begin{center}
 \begin{tabular}{||c|  c | c | c | c |c |c ||}
 \hline
    Type  & d=1 & d=2 & d=3  \\ [0.5ex]
 \hline\hline
  AI     &  &    &  \\
  \hline
 BDI      & $KR(TX)   $&   &  \\
 \hline
 D    & $KR^{-2}(TX)$ &  $KR(TX)  $  &  \\
  \hline
 DIII     & $KR^{-2}(TX)  $ & $KR^{-2}(TX)$ &   $KR(TX)   $  \\  
 \hline
  AII    &  & $KR^{-2}(TX)  $ &  $KR^{-2}(TX)$   \\
  \hline
CII   & $KQ(TX)  $ &  & $KR^{-2}(TX)  $ \\
 \hline
 C     &  & $KQ(TX)  $ &   \\
  \hline
 CI    &   &   & $KQ(TX)  $ \\ [1ex]
 \hline
\end{tabular}
\end{center}

\begin{proof}[Proof of Theorem \ref{idmap}]
  From the table of topological invariants in KO-theory of a point, it is clear that the diagonal map is the identity map.
  Let us look at the above table of sources, and it is enough to check the first nonzero diagonal with entries $KR(TX) = KR^{0,0}(TX)$.
  In the direction of rows, if the index increases by one, the symmetry class number $k$ is increased by one; similarly
  if the column index increases by one, one more spatial dimension is added. For example, from the first $KR(TX)$ of type BDI  to the second $KR(TX)$ of type D is given by a shift from
  $KR^{0,0}$ to $KR^{1,1}$, which is an isomorphism by the $(1,1)$-periodicity of KR-theory. This argument can be applied to any KR-group $KR^{k, \ell}$ in this table,
  where $k$ is the symmetry class number and $\ell$ is the column number (i.e., spatial dimension), both starting from $0$.
  
  The symmetry class number $k = n-1$ can be computed by the chiral symmetry (CS) column,
  which has a deep connection with the real Clifford algebra $C\ell_{0,n}$. In Kitaev's original paper \cite{K09}, K-theory comes from the nontrivial Clifford module extension problem
  from $C\ell_{0,n-1}$-module to $C\ell_{0,n}$-module.
  Each number $0$ or $1$ in a pair $(0,1)$ in the column of CS contributes to the total sum by $2^0 = 1$ or $2^1= 2$. For example, for type AII, read from the top to the row with AII, type AII has the code 01010 within three pairs,
  so 110 and $2^1 + 2^1 +2^0=2+2 +1 =5$ give the symmetry class number $5-1 = 4$.
    
\end{proof}

By the Thom isomorphism in KR-theory, we have the table of bulk KR-theory as the sources.
\begin{center}
 \begin{tabular}{||c|  c | c | c | c |c |c ||}
 \hline
    Type  & d=1 & d=2 & d=3  \\ [0.5ex]
 \hline\hline
  AI     &  &    &  \\
  \hline
 BDI      & $KR^{-1}(X)   $&   &  \\
 \hline
 D    & $KR^{-3}(X)$ &  $KR^{-2}(X)   $  &  \\
  \hline
 DIII     & $KR^{-3}(X) $ & $KR^{-4}(X)$ &   $KR^{-3}(X)   $  \\  
 \hline
  AII    &  & $KR^{-4}(X)  $ &  $KR^{-5}(X)$   \\
  \hline
CII   & $KR^{-5}(X)  $ &  & $KR^{-5}(X)  $ \\
 \hline
 C     &  & $KR^{-6}(X)  $ &   \\
  \hline
 CI    &   &   & $KR^{-7}(X)  $ \\ [1ex]
 \hline
\end{tabular}
\end{center}

The following table collects how one can compute the topological invariant by a topological index. 
By convention,  we use $ch_1(p)$ to denote the first Chern character of a vector bundle represented by a projection $p$, 
and $ch_i(w)$ ($i= 1, 3$) to denote the odd Chern character of the transition function $w$ in an odd K-theory.
In addition, we use $ch_i^{(2)}(w)$ to denote a $\mathbb{Z}_2$-valued odd topological index, and $ch_1^{(2)}(w)_{2 \rightarrow 1}$ to denote 
a $\mathbb{Z}_2$-valued odd topological index after the dimensional reduction from 2d to 1d.

\begin{prop}
   The table of topological indices can be used to compute the bulk invariants of topological insulators and superconductors. 
\end{prop}

\begin{center}
 \begin{tabular}{||c|  c | c | c | c |c |c ||}
 \hline
    Type  & d=1 & d=2 & d=3  \\ [0.5ex]
 \hline\hline
  AI     &  &    &  \\
  \hline
 BDI      & $ch_1(w)$ &   &  \\
 \hline
 D    & $ch_1^{(2)}(w)$ &  $ch_1(p)  $  &  \\
  \hline
 DIII     & $ch_1^{(2)}(w)$ & $ch_1^{(2)}(w)_{2 \rightarrow 1}$  &   $ch_3(w)  $  \\  
 \hline
  AII    &  & $ch_1^{(2)}(w)_{2 \rightarrow 1}$ &   $ch_3^{(2)}(w)  $  \\
  \hline
CII   & $ch_1(w)$ &  &  $ch_3^{(2)}(w)  $ \\
 \hline
 C     &  & $ch_1(p)  $ &   \\
  \hline
 CI    &   &   &  $ch_3(w)  $ \\ [1ex]
 \hline
\end{tabular}
\end{center}


\nocite{*}
\bibliographystyle{plain}
\bibliography{Bott}

\end{document}